\documentclass[11pt]{article}
\usepackage[totalwidth=13.0cm,totalheight=20.0cm]{geometry}
\usepackage{latexsym,amsthm,amsmath,amssymb, amsbsy,url}
\usepackage[ruled, linesnumbered]{algorithm2e}
\usepackage{authblk}
\usepackage{tikz}
\usepackage[retainorgcmds]{IEEEtrantools}

\usetikzlibrary{decorations.pathreplacing}

\newtheorem{lemma}{Lemma}
\newtheorem{corollary}{Corollary}
\newtheorem{definition}{Definition}
\newtheorem{theorem}{Theorem}

\newtheorem{proposition}{Proposition}

\title{The Euler and Chinese Postman Problems on 2-Arc-Colored Digraphs\footnote{Research of GG was partially supported by Royal Society Wolfson Research Merit Award. Research of RL was partially supported by NNSFC under no. 11401353 and TYAL of Shanxi. Research of BS was supported by China Scholarship Council.}}
\author[1]{Bin Sheng}
\author[2]{Ruijuan Li}
\author[1]{Gregory Gutin\footnote{Corresponding author. Email: g.gutin@rhul.ac.uk}}
\affil[1]{Department of Computer Science, Royal Holloway\\ University of London, UK}
\affil[2]{School of Mathematical Sciences, Shanxi University, P.R. China}
\begin{document}

\maketitle
\begin{abstract}
The famous Chinese Postman Problem (CPP) is polynomial time solvable on both undirected and directed graphs. Gutin et al. [Discrete Applied Math 217 (2016)]   generalized these results by proving that CPP on $c$-edge-colored graphs is polynomial time solvable for every $c\geq 2$. In CPP on weighted edge-colored graphs $G$, we wish to find a minimum weight properly colored closed walk containing all edges of $G$ (a walk is properly colored if every two consecutive edges are of different color, including the last and first edges in a closed walk). In this paper, we consider CPP on arc-colored digraphs (for properly colored closed directed walks), and provide a polynomial-time algorithm for the problem on weighted 2-arc-colored digraphs. This is a somewhat surprising result since it is NP-complete to decide whether a 2-arc-colored digraph has a properly colored directed cycle [Gutin et al., Discrete Math 191 (1998)]. 
To obtain the polynomial-time algorithm, we characterize 2-arc-colored digraphs containing properly colored Euler trails. 
\end{abstract}

\section{Introduction}\label{sec:intro}
In this paper, we study the Euler and Chinese Postman Problems on edge-colored  digraphs. (To facilitate reading of Section \ref{sec:intro}, we use the term edge for both directed and undirected graphs. However, we switch to the term arc for digraphs starting from Definition \ref{def:cpp} as the term arc is widely used in the digraph literature, cf. \cite{bang2009digraphs}.) A (directed or undirected) multigraph $G$ is called {\em $c$-edge-colored} if each edge is assigned a color from $[c] = \{1, 2, \ldots, c\}$. Note that the $c$-edge coloring can be arbitrary, not necessarily proper. Most of research on edge-colored multigraphs is related to properly colored walks. A \textit{properly colored (PC)} walk is a walk in which no two consecutive edges have the same color, including the last and first edges in a closed walk. (Since henceforth in digraphs we will deal only with directed walks, trails, cycles and paths, we will omit adjective ``directed'' in such cases.)

PC walks in edge-colored undirected multigraphs are of interest in many applications.
For instance, in genetic and molecular biology \cite{pevzner18computational,szachniuk2014orderly,szachniuk2009assignment}, where Hamilton cycles/Euler trails with certain color pattern are to be found, in transportation and connectivity problems \cite{amaldi2011minimum,galbiati2014minimum,gamvros2006satellite}, where reload costs associated with each pair of colors for incident edges are considered, in design of printed circuit and wiring boards \cite{tseng2010obstacle}, and in channel assignment in wireless networks \cite{ahuja2010algorithms,sankararaman2014channel}. In \cite{xu2010using,xu2009matrix}, edge-colored directed multigraphs are used to model conflict resolution. The graph model can be viewed as a game theory-related tool that can assist negotiators with the strategic aspects of a negotiation. An edge-colored directed multigraph of a conflict allows for an extensive analysis of the possible strategic interactions among decision makers or agents.

There are many positive algorithmic results on PC walks in edge-colored graphs, for a detailed survey, we refer interested readers to Chapter 16 of \cite{bang2009digraphs} for pre-2009 literature, and to, e.g., \cite{abouelaoualim2010cycles,fujita2011properly,gutin2017chinese,gutin2017odd,lo2014dirac,lo2014edge} for later publications. Unfortunately, most problems turn out to be much harder for edge-colored digraphs than for edge-colored undirected graphs.  In particular, it was proved to be NP-hard to decide whether there is a PC cycle
in a given 2-edge-colored digraph \cite{gutin1998note}. In comparison, the problem of deciding the existence of PC cycle in a $c$-edge-colored undirected graph is polynomial-time solvable for every $c\ge 2$ \cite{yeo1997note}. 
In \cite{gourves2013complexity}, the authors proved that it is NP-hard to decide whether there is a PC path between two given vertices even in a $c$-edge-colored planar digraph which contains no PC cycle for $c=\Omega(|V(G)|)$. In the same paper, it was proved that deciding the existence of a PC cycle through a given vertex in a $c$-edge-colored tournament $T$ is NP-hard. In addition, deciding whether $T$ has a PC $s$-$t$ path or a PC Hamilton $s$-$t$ path is NP-complete. However, there were a couple of positive results proved in \cite{gourves2013complexity}, in particular, it was proved to be polynomial-time solvable to decide whether an edge-colored digraph contains a PC closed trail and to compute the maximum number of edge disjoint PC trails between any two vertices.

 As PC walks in edge-colored undirected graphs (and thus in edge-colored digraphs) are generalizations of walks in both undirected and directed graphs, we would like to extend results on walks in undirected and directed graphs to edge-colored undirected and directed graphs whenever it is possible. For example, the Chinese Postman Problem is polynomial-time solvable on both undirected and directed multigraphs \cite{bernhard2008combinatorial,edmonds1973matching}. This result was recently generalized to weighted edge-colored undirected graphs $G$ in \cite{gutin2017chinese}: we can compute a minimum weight PC walk in $G$ in polynomial time.  While we do not know whether this result can be further extended to all $c$-edge-colored digraphs for $c\ge 2$, we will show that it can be done for $c=2$. To obtain our result, we first prove a characterization of {\em PC Euler 2-edge-colored digraphs}, i.e.  2-edge-colored digraphs $G$ which have a PC Euler trail. (Recall that a {\em trail} is a walk without repetition of edges and a trail is {\em Euler} if it is closed and contains all edges of $G$.) 
Our characterization is of independent interest and generalizes a characterization of Kotzig \cite{kotzig1968moves} of PC Euler $c$-edge-colored undirected graphs for $c=2$. Note that our characterization requires a new notion of PC trail-connectivity introduced in Section \ref{sec:main}.

Let us give a formal definition the Chinese Postman Problem studied in this paper.

\begin{definition}[CPP-ACD]\label{def:cpp}
 Given a $c$-arc-colored digraph $G$ for $c\ge 2$, with non-negative weights on its edges, find a {PC closed walk} in $G$ which traverses each arc of $G$ at least once and has the minimum weight among all such walks.
\end{definition}

CPP-ACD with $c=2$ will be denoted by CPP-2ACD.
A {PC closed walk} in $G$ which traverses each arc of $G$ at least once is called a \textit{feasible solution} for CPP-ACD on $G$. A feasible solution with minimum weight is called an \textit{optimal solution} for CPP-ACD. Observe that for a given arc-colored digraph $G$, it is possible that there is no feasible solution for CPP-ACD on $G$. For instance, any digraph with at least two arcs in which all arcs have the same color has no solution for CPP-ACD.

We present our results in the following order. We will first provide necessary and sufficient conditions for a $c$-arc-colored digraph to have a feasible solution for CPP-ACD. Here we already need the new notion of PC trail-connectivity. Then we introduce another notion, that of PC circuits, using which we will show how to find a PC Euler trail in a 2-edge-colored digraph if it contains one.  At last we will prove that CPP-2ACD is polynomial time solvable.

It is unclear whether CPP-ACD, in all its generality, is polynomial-time solvable as 
some generalizations of the Chinese Postman Problem (CPP) were proved to be NP-hard, such as CPP on mixed multigraphs \cite{papadimitriou1976complexity} and $k$-CPP on both undirected and directed graphs \cite{gutin2013parameterized,thomassen1997complexity}. Parameterized algorithms is a powerful tool to tackle NP-hard problems, therefore there are several parameterized studies for these hard variants already, cf. \cite{fernandes2009minimum,gutin2015parameterized:2,gutin2014parameterized:3,gutin2015structural,gutin2013parameterized}. For a systematic introduction on classical and parameterized study of CPP and its generalizations, see the comprehensive survey \cite{van2014complexity}.

\section{Notation and Terminology}\label{sec:not}
For most of the graph theoretical concepts used in this paper, we follow the notation and terminology in \cite{bang2009digraphs,bollobas2013modern}.

Given a digraph $G$, if $uv\in A(G)$, then we call the arc $uv$ an \textit{incoming arc} of $v$ and an \textit{outgoing arc} of $u$; both $u$ and $v$ are called the \textit{end-vertices} of $uv$. More specifically, $v$ is the \textit{head} of $uv$ and $u$ is the \textit{tail} of $uv$.
We say a digraph $G$ is \textit{strongly connected} if there is a path from $u$ to $v$ and a path from $v$ to $u$ for any two vertices $u,v$ in $G$.
In a weighted digraph, the weight of a walk is the total weights of the arcs in the path. 
%A trail $T = v_1f_{1}v_2\ldots v_{p-1}f_{p-1}v_{p}$ is a sequence of alternating vertices and arcs, such that $f_{i}$ is an arc from $v_{i}$ to $v_{i+1}$ for every %$i\in[p-1]$ and $f_i \neq f_j$ if $i\neq j$. Trail $T$ is \textit{closed} (\textit{open}, respectively) if $v_{1}=v_{p}$ $(v_{1}\neq v_{p}$, respectively). Arc $f_1$ %($f_{p-1}$, respectively) is called the first (last, respectively) arc of $T$.
%An \textit{out-branching} $T$ of a connected digraph $G$ is a spanning oriented tree rooted at some vertex $r$ such that $r$ has in-degree 0, and every other vertex in $T$ has in-degree exactly 1. An out-branching of a connected undirected graph $G$ is defined as the out-branching of the double orientation of $G$, which is obtained by replacing each edge $uv\in E(G)$ by the pair of arcs $uv, vu$.

Let $f=xy$ be an arc in an weighted directed multigraph $G$. The operation of \textit{double subdividing} $f$ replaces $f$ with an weighted directed path $P_f=xu_{xy}v_{xy}y$ from $x$ to $y$ with three arcs such that the weight of $P_f$ equals to that of $f$.
Let $v$ be a vertex of a digraph $G$. By \textit{splitting} the vertex $v$ we mean adding a vertex $v'$ and a new arc $vv'$ and then replacing each arc $vw$ with the arc $v'w$. And we say vertex $v$ {\em is split} into the arc $vv'$.

%A graph $H = (V', A')$ is a \textit{subgraph} of $G$ if $V'\subseteq V$, $A' \subseteq A$ and each arc $f\in A'$ has weight $\omega_{G}(f)$ and color $\phi_{G}(f)$. We say a subgraph $H$ of $G$ is induced by the arc set $A'$ if $A'\subseteq A(G)$ and $V'=\{u\in V(G)| u \in V(A')\}$.

 In an arc-colored digraph $G=(V(G),A(G))$ with an arc coloring $\phi: A(G)\rightarrow [c]$, let $T = v_1v_2\ldots v_{p-1}v_{p}$ be a trail in $G$. We say $T$ is a trail \textit{starting at vertex $v_1$ with arc $v_1v_2$} and \textit{ending at vertex $v_p$ with arc $v_{p-1}v_p$}; sometimes we call $T$ a trail \textit{starting  with color} $\phi(v_1v_2)$ and \textit{ending  with color} $\phi(v_{p-1}v_p)$.
Recall that a trail $T$ is {\em properly colored} if $\phi(v_{i}v_{i+1})\neq \phi(v_{i+1}v_{i+2})$ for any $i\in [p-2]$, and $\phi(v_{1}v_2)\neq \phi(v_{p-1}v_p)$ if $T$ is a closed trail.
For subgraphs $G_1,G_2$ of $G$, a PC trail $T = v_1v_2\ldots v_{p-1}v_{p}$ \textit{switches} from $G_1$ to $G_2$ via the vertex $v_2$, if $v_1v_2\in A(G_1)$ and $v_2v_3\in A(G_2)$.

Let $G=(V,A)$ be a $c$-arc-colored directed multigraph, whose arc coloring is denoted by $\phi: A(G)\rightarrow [c]$. For a vertex $v \in V(G)$ and color $i\in[c]$, let $d_{i, G}^{-}(v)$ be the number of incoming arcs of $v$ colored with color $i$ in $G$, and $d_{i,G}^{+}(v)$ be the number of outgoing arcs of $v$ colored with color $i$ in $G$. Then $d^{-}_{G}(v)=\Sigma_{i \in [c]}d_{i,G}^{-}(v)$ is the in-degree of $v$ in $G$, and  $d^{+}_{G}(v)=\Sigma_{i \in [c]}d_{i, G}^{+}(v)$ is the out-degree of $v$ in $G$. We write $i \in \phi^{-}(v)$ if there is an arc $uv \in G$, such that $\phi(uv) = i$ and similarly $i \in \phi^{+}(v)$ if there is an arc $vw \in G$, such that $\phi(vw) = i$, for $i\in [c]$. We write $u \in N^{-}_{i}(v)$ if $uv \in A(G)$ and $\phi(uv) = i$, and we write $w \in N^{+}_{i}(v)$ if $vw \in A(G)$ and $\phi(vw) = i$.

The \textit{underlying graph} of an arc-colored digraph $G$ is the undirected graph $H$ where $V(H)=V(G)$, and $uv \in E(H)$ if and only if $uv \in A(G)$ or $vu\in A(G)$.
An arc-colored digraph is \textit{connected} if its underlying undirected graph is connected. And so when we talk about a connected component in an arc-colored digraph, we mean its connected component in the underlying graph.  In an arc-colored directed multigraph $G$, the \textit{multiplicity} of an arc $f$ is the number of arcs that have same tail, head and color; we denote the multiplicity of arc $f$ by $\mu(f)$.
{Given an arc-colored digraph $G$ and a subgraph $H$ of $G$, $G-H$ ($G+H$) is the arc-colored digraph obtained by deleting (adding) all the arcs of $H$ from $G$.} We say $H$ is a \textit{PC Euler subgraph} of $G$ if $H$ is a subgraph of $G$ and $H$ is PC Euler.

% In an arc-colored digraph, the \textit{reload cost} is defined as the cost when a trail in the graph switch from an arc of one color to an arc of another color. The reload costs of a $c$-arc-colored digraph are determined by a $c\times c$ matrix $[r_{ij}]_{c\times c}$, whose entry $r_{i,j}$ corresponds to the reload cost of switching  from an arc of color $i$ to an arc color $j$, where we allow $i=j$.

For two weighted arc-colored  directed multigraphs $G_1$ and $G_2$, we define the \textit{union} of $G_1$ and $G_2$ to be the weighted arc-colored  directed multigraph $H$, such that $V(H)= V(G_1)\cup V(G_2)$, and $A(H)$ is the union of the two multi-sets $A(G_1)$ and $A(G_2)$, i.e., $\mu_{H}(xy) = \mu_{G_1}(xy)+\mu_{G_2}(xy)$, for any $x,y\in V(H)$.
The arc coloring of $H$, $\phi: A(H) \rightarrow [c]$ inherits from $\phi_{1}: A(G_1) \rightarrow [c]$ and $\phi_{2}: A(G_2) \rightarrow [c]$, that is, $\phi(xy) = \phi_{i}(xy)$ if $xy$ is a copy comes from $A(G_i)$, for $i\in [2]$. The arc weight of $H$, $\omega: A(H) \rightarrow \mathbb{R}_{\geq 0}$ is determined by $\omega_{1}: A(G_1) \rightarrow \mathbb{R}_{\geq 0}$ and $\omega_{2}: A(G_2) \rightarrow \mathbb{R}_{\geq 0}$, that is, $\omega(xy) = \omega_i(xy)$ if $xy$ is a copy comes from $A(G_i)$, for $i\in [2]$.

\section{CPP-ACD}\label{sec:main}

In a $c$-arc-colored digraph $G$, a vertex $v$ is called \textit{color-balanced}  if $d^-(v)=d^+(v)$, and $d^{+}_{i,G}(v)\leq \sum_{j\neq i\in[c]} d^{-}_{j,G}(v)$ and $d^{-}_{i,G}(v)\leq \sum_{j\neq i\in[c]} d^{+}_{j,G}(v)$ for any $i\in[c]$.
We say that $G$ is \textit{color-balanced} if every vertex in $G$ is color-balanced.
Note that, in a $2$-arc-colored digraph $G$, a vertex $v$ is \textit{color-balanced} if and only if $d^{+}_{i,G}(v)= d^{-}_{3-i,G}(v)$, for $i\in[2]$.

It is easy to see that in our study of the Euler and Chinese Postman Problems, we may restrict ourselves to arc-colored digraphs rather than arc-colored directed multigraphs: If there is an arc $f$ with multiplicity $\mu(f)>1$, it suffices to double subdivide each copy $g$ of $f$ (obtaining a path $P_g$) and assign color $\phi(f)$ to the first and third arcs of $P_g$ and color $3-\phi(f)$ to middle arc. The weight of $g$ can be arbitrarily distributed to the arcs of $P_g.$

An arc-colored digraph $G$ is \textit{PC trail-connected}, if there is a PC trail starting with arc $f_1$ and ending with arc $f_2$, for any pair of arcs  $f_1, f_2$ in $G$.

\subsection{Feasibility in CPP-ACD}\label{sec:tc}

Let us start from the following simple yet useful assertion.

\begin{proposition}\label{proposition}
 Let $G$ be a $c$-arc-colored digraph with at least 2 arcs. If $G$ is PC trail-connected, then we have the following.
 \begin{enumerate}
 \item $G$ is strongly connected.
 \item For any vertex $v\in V(G)$ and color $i\in [c]$, if $d^{-}_{i,G}(v) > 0$, then there exists $j\neq i\in[c]$, such that $d^{+}_{j,G}(v) > 0$. Similarly if $d^{+}_{i,G}(v) > 0$, then there exists $j\neq i\in[c]$, such that $d^{-}_{j,G}(v) > 0$.
 \end{enumerate}
\end{proposition}
\begin{proof}
If $G$ is PC trail-connected, then there is a PC trail starting with arc $f_1$ and ending with arc $f_2$, for any pair of arcs  $f_1, f_2$ in $G$. It is easy to see that $G$ is strongly connected.

  Moreover, if there is an arc $uv$ in $G$ with $\phi(uv) = i$, there must be an arc $vw$ with $\phi(vw) = j \neq i$, otherwise, arc $uv$ cannot reach any other arc via a PC trail. Similarly, if there is an arc $vw$ in $G$ with $\phi(vw) = i$, there must be an arc $uv$ with $\phi(uv) = j\neq i$, otherwise, no arc can reach $vw$ via a PC trail.
\end{proof}

 We now prove a necessary and sufficient condition for a $c$-arc-colored digraph to have a solution for CPP-ACD, which is the reason we introduced the concept of PC trail-connectivity.
 
 \begin{lemma}\label{lemma: solutionCondition2}
 For a $c$-arc-colored digraph $G$, there is a feasible solution for CPP-ACD on $G$ if and only if $G$ is PC trail-connected.
 \end{lemma}
 \begin{proof}
  On the one hand, if CPP-ACD has a solution on $G$, then there is a PC Euler directed multigraph $G'$, which is obtained by adding copies of some arcs in $G$. For any pair of arcs $f_1, f_2$ in $G$, since $G'$ has a PC Euler trail, $G'$ has a PC trail from $f_1$ to $f_2$ and a PC trail from $f_2$ to $f_1$. Note that a PC trail in $G'$ corresponds to a PC walk in $G$. Thus, $G$ has a PC walk from $f_1$ to $f_2$ and a PC walk from $f_2$ to $f_1$. It is not hard to see that the PC walks of $G$ can be shortened to PC trails with the same first and last arcs. Thus, $G$ is PC trail-connected.

% Let $T$ be a minimal PC walk from $f_1$ to $f_2$ in $G$. {We claim that no arc in $G$ appears more than once in $T$, which implies that $T$ is a PC trail in $G$.} Suppose the claim is wrong, and let $f$ be an arc which appears more than once in $T$; then we know  $T = (f_1, T'_1, f,T'_2,f,T'_3,f_2)$, for some PC walks $T'_1, T'_2$ and $T'_3$ in $G$. Note that $T' = (f_1, T'_1, f,T'_3,f_2)$ is a PC walk from $f_1$ to $f_2$ contained in $T$, a contradiction to the assumption that $T$ is minimal. Therefore, no arc appears more than once in $T$ and so $T$ is a PC trail from $f_1$ to $f_2$ in $G$. It follows that $G$ is PC trail-connected.

 On the other hand, we need to prove that if $G$ is PC trail-connected, then there is a feasible solution for CPP-ACD on $G$. To show this, we explicitly construct a PC closed walk in $G$ that contains each arc of $G$ at least once.
 Let $f_1,f_2$ be distinct arcs of $G$. Since $G$ is trail-connected, there is a PC trail $T_1$ from $f_1$ to $f_2$. Let $f_3$ be an arc of $G$ not contained in $T_1$. There is a PC trail $T_2$ from $f_2$ to $f_3$ in $G$. Continue this way by choosing an arc $f_i$ not contained in $\bigcup_{j\in [i-2]}T_j$ and finding a PC trail $T_{i-1}$ from $f_{i-1}$ to $f_i$. We will end the procedure when all arcs of $G$ are covered by the PC trails. Suppose when the procedure stops $i=t$. 
 
 Note that $W=\bigcup_{i\in [t-1]}T_{i}\setminus \{f_j: 2\leq j \leq t-1\}$ is a trail starting with $f_1$ and ending with $f_t$ which is properly colored unless it is closed and $\phi(f_1)= \phi(f_t)$.
%  We can merge $T_i, i\in [t-1]$ into a PC walk $W$, by removing all first arcs of the PC trails apart from $f_1$ and get a PC trail from $f_1$ to $f_t$, where $f_t$ is the last arc of the last PC trail.  
 If $W$ is closed and $\phi(f_1)\neq \phi(f_t)$, then $W$ is a PC closed walk of $G$ containing all arcs of $G$. Otherwise, there is an arc $f$ in $G$ whose head is the tail of $f_1$ and $\phi(f)\neq \phi(f_1)$ by  Proposition \ref{proposition}. Let $T_t$ be a PC trail from $f_t$ to $f$. Then $W'=W\cup T_t\setminus\{f_t\}$ is a PC closed walk in $G$ containing all arcs of $G$. 
  \end{proof}

\subsection{PC Euler 2-ACD}

In this section, our main result is the following theorem, which provides necessary and sufficient conditions for a 2-arc-colored digraph to be PC Euler, and gives a way to find a PC Euler trail if it exists.

\begin{theorem}\label{theorem:EulerCondition}
Let $G$ be a 2-arc-colored directed multigraph. Then $G$ is PC Euler if and only if $G$ is color-balanced and PC trail-connected. Moreover, a PC Euler trail of $G$ can be constructed in polynomial time if it exists.
\end{theorem}

It is easy to decide whether $G$ is color-balanced. Corollary \ref{corollary: sufficientCondition} proved in the next subsection, allows us to decide whether $G$ is PC trail-connected.

It is not clear how to prove Theorem \ref{theorem:EulerCondition} by induction, as it is possible that the remaining part of $G$ after deleting some PC Euler subgraphs is not PC trail-connected.

We introduce the following notion of PC circuits, which is a counter part of cycles in undirected and directed graphs.

\begin{definition}
A subgraph $C$ of $G$ is called a PC circuit, if it is PC Euler and $d^{+}_{1,C}(v) \leq 1$, $d^{+}_{2,C}(v) \leq 1$ hold for any vertex $v \in V(C)$.
\end{definition}

\begin{lemma}
Let $G$ be a 2-arc-colored directed multigraph. If $C$ is a minimal PC Euler subgraph of $G$, then for any vertex $v\in V(C)$, $d^{+}_{1,C}(v) \leq 1$ and $d^{+}_{2,C}(v) \leq 1$.
\end{lemma}
\begin{proof}
 Suppose on the contrary, there is a vertex $v\in V(C)$ with $d^{+}_{1,C}(v) \geq 2$. Observe that a PC trail starting at $v$ with an outgoing arc colored 1 and ending at $v$ with an incoming arc colored 2 is a PC Euler subgraph.
 Consequently, there is a proper subgraph of $C$ which is PC Euler, a contradiction to the assumption that $C$ is a minimal PC Euler subgraph.
\end{proof}

\begin{corollary}\label{corollary}
Let $G$ be a 2-arc colored directed multigraph. If $C$ is a minimal PC Euler subgraph of $G$, then $C$ is a PC circuit.
\end{corollary}

\begin{definition}
A PC circuit $C$ is {\em bad} in $G$, if there is a connected component $D$ in $G-C$ such that $\max\{d^{+}_{1, C\cup D}(v),d^{+}_{2, C\cup D}(v)\}=1$, for any vertex $v \in V(C)\cap V(D)$. We say a PC circuit is {\em good} in $G$ if it is not bad in $G$.
\end{definition}

%Observe that a PC circuit is good if and only if there is at least one good vertex in it.

\begin{lemma}\label{lemma: noBad}
Let $G$ be a 2-arc-colored directed multigraph. If $G$ is PC trail-connected, then there is no PC circuit which is bad in $G$.
\end{lemma}
\begin{proof}

Suppose on the contrary, there is a PC circuit $C$ which is bad in $G$. As $G$ is PC trail-connected and thus strongly connected by Proposition \ref{proposition}, there must be a vertex $v\in V(C)$ such that $d_{1,G}^{+}(v)+d_{2,G}^{+}(v)>2$. Otherwise $G=C$ is simply a PC circuit, in which case, $C$ can not be bad in $G$, a contradiction.

As $C$ is bad in $G$, then by definition, there is a connected component $D$ in $G-C$ such that for any vertex $v \in V(C)\cap V(D)$, $\max\{d^{+}_{1, C\cup D}(v),d^{+}_{2, C\cup D}(v)\}=1$. Consequently, there is no PC trail in $G$ which switches from $C$ to $D$ via $v$, since such a PC trail implies that there are two arcs $uv\in A(C)$ and $vw\in A(D)$ such that $\phi(uv)\neq \phi(vw)$. As $G$ and $C$ are both color-balanced, $D$ is also color-balanced. The fact that both $C$ and $D$ are color-balanced and $\phi(uv)\neq \phi(vw)$ implies that $d^{+}_{1, C\cup D}(v)\geq 2$ or $d^{+}_{2, C\cup D}(v)\geq 2$, which is not possible since $\max\{d^{+}_{1, C\cup D}(v),d^{+}_{2, C\cup D}(v)\}=1$. Hence for any two arcs $u_1v_1\in A(C)$ and $u_2v_2\in A(D)$, there is no PC trail starting with arc $u_1v_1$ and ending with arc $u_2v_2$ in $G$, which is a contradiction to the fact that $G$ is PC trail-connected.
\end{proof}

\begin{lemma}\label{lemma: sufficiency}
Let $G$ be a 2-arc-colored directed multigraph. If $G$ is color-balanced, then $G$ can be decomposed into a set of PC circuits in polynomial time.
\end{lemma}

\begin{proof}
Let $G$ be a 2-arc-colored directed multigraph which is color-balanced. We will give a polynomial-time algorithm to decompose $G$ into minimal PC Euler subgraphs, which are PC circuits by Corollary  \ref{corollary}.

Initially label all arcs in $G$ ``non-traversed''. We find a minimal PC Euler subgraph $C$ of $G$ in the following way. Start by setting $R=uv$, which is some non-traversed arc. Then keep adding arcs into $R$ along a PC trail using only the ``non-traversed'' arcs. The procedure stops when the first time $R$ contains a PC Euler subgraph, and then we denote a PC Euler subgraph in $R$ which contains the last added arc by $C$.  We then change the labels of all arcs in $C$ to ``traversed''.

Note that we can check whether $R$ contains a PC Euler subgraph in polynomial time. Each time we add an arc $xy$ into $R$ such that vertex $y$ is already in $V(R)$, let us check whether $y$ satisfies the following condition %$$\{\phi(xy)\}\cup \phi^{+}_{R}(y) = \{1,2\}$$, which we call Condition ($*$). 

\begin{equation}\label{equation0}
\{\phi(xy)\}\cup \phi^{+}_{R}(y) = \{1,2\}
\end{equation}

If (\ref{equation0}) does not hold for $y$, then there is no PC Euler subgraph in $R$.
If (\ref{equation0}) holds for $y$, then $R$ contains a PC Euler subgraph, as $R$ contains a PC closed trail starting and ending at $y$. Note that the PC trail starting at $y$ with color $3-\phi(xy)$ and ending at $y$ with arc $xy$ induces a minimal PC Euler subgraph.

As $C$ is a PC Euler subgraph of $G$, the subgraph of $G$ induced by all the non-traversed arcs remains color-balanced.
Consequently, the above procedure of finding minimal PC Euler ``non-traversed'' subgraphs continues, until there is no non-traversed arc in $G$. In this way, we decompose $G$ into minimal PC Euler subgraphs, which are PC circuits.
\end{proof}

Let $G$ be a 2-arc-colored digraph which is PC trail-connected. For two PC circuits $C_1$ and $C_2$ in $G$, such that $V(C_1) \cap V(C_2) \neq \emptyset$, we say a vertex $v\in V(C_1)\cap V(C_2)$ is \textit{good} if $d^{+}_{1, C_1 \cup C_2}(v)\geq 2$ or $d^{+}_{2, C_1\cup C_2}(v) \geq 2$.

\smallskip
The following lemma shows that $G$ is PC Euler if $G$ can be decomposed into PC circuits that are good in $G$. This is the reason we care about whether a PC circuit is good or bad in $G$.

\begin{lemma}\label{lemma: goodDecomposition}
Let $G$ be a 2-arc-colored directed multigraph which is PC trail-connected. Suppose that $G$ can be decomposed into a set of PC circuits, i.e., $G = C_1 \cup C_2 \cup \ldots \cup C_s$ for some positive integer $s$. If for each $i\in [s]$, $C_i$ is a good PC circuit in $G$, then $G$ is PC Euler.
\end{lemma}
\begin{proof}

We construct an auxiliary graph $H$, such that $V(H)=\{v_i| i\in [s]\}$, and there is a one to one correspondence between vertex $v_i\in V(H)$ and the good PC circuit $C_i$ in $G$ for each $i\in[s]$. We add an edge between two vertices $v_i, v_j$ in $H$ if and only if there exists a good vertex  $y\in V(C_i)\cap V(C_j)$. By the definition of PC circuits, $\max\{ d^{+}_{1, C_i}(y),$ $ d^{+}_{2, C_i}(y), d^{+}_{1, C_j}(y), d^{+}_{2, C_j}(y)\}\leq 1$, therefore $d^{+}_{1, C_i}(y)= 1$, $d^{+}_{1, C_j}(y)= 1$ or $d^{+}_{2, C_i}(y) = 1 $, $d^{+}_{2, C_j}(y) = 1 $.  In both cases, we can switch from $C_i$ to $C_j$ via the vertex $y$, and switch back to $C_i$ via $y$ after a PC Euler trail of $C_j$, and vice versa.

Now, to see that $G$ is PC Euler, we first prove that graph $H$ is connected.
For any two arcs $f_1, f_t$ where $f_1\in A(C_i)$ and $f_t\in A(C_j)$, there is a PC trail $T$ from $f_1$ to $f_t$, as $G$ is PC trail-connected. Let $T = f_1f_2\ldots f_{t-1}f_t$. For any two consecutive arcs $f_{k}$ and $f_{k+1}$ in $T$ where $f_k\in A(C_{r_1})$ and $f_{k+1}\in A(C_{r_2})$, $r_1\neq r_2 \in[s]$, we know that $\phi(f_k) \neq \phi(f_{k+1})$. Moreover, as both $C_{r_1}$ and $C_{r_2}$ are color-balanced, it follows that the head of $f_k$ is a good vertex in $V(C_{r_1})\cap V(C_{r_2})$. Consequently $v_{r_1}v_{r_2}\in E(H)$ by the construction of $H$. The fact that $T$ starts with arc $f_1$ in $C_i$ and ends with arc $f_t$ in $C_j$ implies there is a walk between $v_i$ and $v_j$ in $H$. The two arcs $f_1$ and $f_t$ are arbitrarily chosen, thus there is a walk between any two vertices in  $H$, and so $H$ is a connected graph.

We construct a PC Euler trail $T_{G}$ of $G$ in the following way. Consider a DFS on $H$ starting at vertex $v_1\in V(H)$. Since $H$ is a connected undirected graph, the predecessor subgraph of the DFS on $H$ is a single depth-first tree $T$. Set initially $T_{G} = T_{1}$, which is a PC Euler trail of $C_{1}$. Label $v_1$ ``visited'' and all other vertices in $H$ ``non-visited''.

We complete the PC Euler trail $T_G$ of $G$ by inserting a PC Euler trail of $C_{i}$ into $T_{G}$ for each $v_i\in V(H)$. We do the insertions following the tree edges in $T$. In each step, we look at an edge in $T$ which has an ``non-visited'' endvertex and a ``visited'' endvertex. Let $v_iv_j$ be such an tree edge in $T$ where $v_i$ is labelled ``visited'' and $v_j$ is labelled ``non-visited''. Then let $T_{i}$ ($T_{j}$, respectively) be a PC Euler trail of $C_i$ ($C_{j}$, respectively). The edge $v_iv_j$ in $T$ implies that there is an edge between $v_i$ and $v_j$ in $H$. Consequently, we can switch between $T_i$ and $T_j$ via some good vertex $y_{ij}\in V(C_i)\cap V(C_j)$. Suppose $x_{ij}y_{ij}$ is an arc in $A(C_i)$. Then we may insert $T_{j}$ into $T_{G}$ at vertex $y_{ij}$ right behind the arc $x_{ij}y_{ij}$. And then we change the label $v_j$ to ``visited''.

Note that after each insertion, the trail $T_{G}$ remains properly colored as $T_j$ starts with color $3-\phi(x_{ij}y_{ij})$ and ends with color $\phi(x_{ij}y_{ij})$. Moreover, the depth-first tree $T$ contains every vertex in $H$, and each vertex (except $v_1$) has exactly 1 parent in $T$. Consequently every vertex would be labelled ``visited'' eventually. And exactly one copy of PC Euler trail of $C_i$ is inserted into $T_G$ for $i\in [s]$. Therefore, when all vertices in $T$ are labelled ``visited'', $T_{G}$ is a PC Euler trail of $G$. Thus, $G$ is PC Euler.
\end{proof}

\noindent\textbf{Proof of Theorem \ref{theorem:EulerCondition}:}
On the one hand, if $G$ is PC Euler, then  any PC Euler trail of $G$ is a feasible solution for CPP-2ACD on $G$, therefore $G$ is PC trail-connected by Lemma \ref{lemma: solutionCondition2}.
Moreover, the existence of a PC Euler trail $T$ in $G$ implies that $d^{-}_{i,G}(v) = d^{+}_{3-i,G}(v)$ holds for any vertex $v\in V(G)$ and any color $i\in[2]$, thus $G$ is color-balanced.

On the other hand, if $G$ is PC trail-connected then there is no PC circuit which is bad in $G$ by Lemma \ref{lemma: noBad}. Moreover, if $G$ is color-balanced, then $G$ can be decomposed into a set of PC circuits by Lemma \ref{lemma: sufficiency}. Consequently, if $G$ is PC trail-connected and color-balanced, then  $G$ can be decomposed into a set of PC circuits, each of which is good in $G$. Then it follows from Lemma \ref{lemma: goodDecomposition} that $G$ is PC Euler. \qed

%\end{proof}

\subsection{Polynomial Time Algorithm for CPP-2ACD}
Let $G$ be a 2-arc-colored weighted digraph which is PC trail-connected.
By Theorem \ref{theorem:EulerCondition}, a 2-arc-colored directed multigraph is PC Euler if and only if it is color-balanced and PC trail-connected. So any color-balanced supergraph of $G$ is a feasible solution for CPP-2ACD on $G$.

In this section, our main result is the following Theorem \ref{theorem: polynomial}, in which we provide a polynomial-time algorithm to find an optimal solution for the Chinese Postman Problem on $G$.

\begin{theorem}\label{theorem: polynomial}
There is a polynomial-time algorithm which solves the CPP-2ACD on any 2-arc-colored weighted digraph $G$.
\end{theorem}

Roughly speaking, in the proof of Theorem \ref{theorem: polynomial}, we first check whether $G$ is PC trail-connected. Then we find an optimal way to add copies of arcs which makes $G$ color-balanced. We introduce the following notion which is essential for both tasks.

A \textit{fixed end-vertex $u$-$v$ trail (FEV $u$-$v$ trail)} in $G$ is a trail from $u$ to $v$ where $u,v\in V(G)$ are the  given end-vertices and we allow $u=v$.
An FEV $v_1$-$v_p$ trail $T = v_1  v_2 \ldots v_{p-1}  v_p$ is PC if $\phi(v_i v_{i+1}) \neq \phi(v_{i+1}v_{i+2})$ for every
$i \in [p - 2]$. Note that we do not require that $\phi(v_{1}v_2)\neq \phi(v_{p-1}v_p)$ when $v_1 = v_p$.  Thus, a PC FEV $v_1$-$v_p$ trail $T$ might not be a PC closed trail if $v_1 = v_p$.

The following lemma helps us to decide whether $G$ is PC trail-connected.

\begin{lemma}\label{lemma:FEVtrail}
Given an arbitrary $c$-arc-colored digraph $G$ for $c\ge 2$, and two vertices $s,t\in V(G)$, we can in polynomial time check the existence of a PC FEV $s$-$t$ trail in $G$, and find a minimum weight PC FEV $s$-$t$ trail if it exists.
\end{lemma}
\begin{proof}
Let $G$ be a $c$-arc-colored weighted digraph with arc weight $\omega: A(G) \rightarrow \mathbb{R}_{\geq 0}$ and arc coloring $\phi: A(G) \rightarrow [c]$. Define an auxiliary digraph $H$ as follows. Let the vertex set of $H$ be $\{(u,v): uv\in A(G)\} \cup \{(x,s),(t,y)\}$. For any two incident arcs $uv, vw\in A(G)$, with $\phi(uv) \neq \phi(vw)$, we add an arc from $(u,v)$ to $(v,w)$ into $H$. We also add an arc from $(x,s)$ to $(s,w)$ for each arc $sw\in A(G)$, and an arc from $(z,t)$ to $(t,y)$ for each arc $zt\in A(G)$. We define the vertex weight function  $\omega: V(H) \rightarrow \mathbb{R}_{\geq 0}$ in the following way. Let $\omega((u,v)) = \omega_{G}(uv)$, and $\omega((x,s))=\omega((t,y)) =0$. This completes the description of $H$. Basically, every vertex in $H$ has the weight of its corresponding arc in $G$, except for $(x,s)$ and $(t,y)$.

We prove that a PC FEV $s$-$t$ trail in $G$ corresponds to a directed path from $(x,s)$ to $(t,y)$ in $H$, with exactly the same weight.

On the one hand, let $P = (x,s)(s,u_1)(u_1,u_2)\ldots $$(u_{p-1},u_p)$$(u_p,t)$$(t,y)$ be a directed path in $H$ from $(x,s)$ to $(t,y)$. By the construction of $H$, there is an arc from $(u,v)$ to $(v,w)$ if and only if $\phi(uv) \neq \phi(vw)$. Thus $T = s u_1 u_2\dots u_{p-1}u_pt$ is a PC FEV $s$-$t$ trail in $G$. Note that the weight of $T$ equals to the weight of $P$.

On the other hand, consider a minimum weight PC FEV $s$-$t$ trail $T' = s v_1 v_2\dots v_q t$ in $G$. Then $Q = (x,s)$$(s,v_1)$$(v_1,v_2)$$\ldots$$(v_{p-1},$$v_p)$$(v_q,t)$$(t,y)$ is a directed path in $H$ from $(x,s)$ to $(t,y)$. By the construction of $H$, the weight of $Q$ equals to the weight of $T'$. It remains to observe that $Q$ is a minimum weight directed path from $(x,s)$ to $(t,y)$, otherwise there is a PC FEV $s$-$t$ trail with weight smaller than that of $T'$, a contradiction.

In polynomial time, we can check if there is a directed path from $(x,s)$ to $(t,y)$ in $H$, and compute one with minimum weight if it exists.
Indeed, in a vertex weighted digraph $J$, for any two vertices $x, y\in V(J)$, computing a minimum weight directed path from $x$ to $y$ can be done in polynomial time.  We construct an weighted graph $J'$ from $J$ by splitting each vertex $v\in V(J)$ into an arc $vv'$, and assign weight $\omega_{J}(v)$ to the arc $vv'$. Moreover, we assign weight 0 to all the other arcs in $J'$. It follows that a minimum weight $x$-$y$ path in $J$ corresponds to a minimum weight $x$-$y'$ path in $J'$, which can be computed in polynomial time using Dijkstra's algorithm. Observe that there is no directed path from $(x,s)$ to $(t,y)$ in $H$ if and only if the Dijkstra's algorithm fails to compute a minimum weight directed path from $(x,s)$ to $(t,y)$ in $H'$.

A minimum weight PC FEV $s$-$t$ trail in $G$ corresponds to a minimum weight directed path from $(x,s)$ to $(t,y)$ in $H$.
By the above arguments, we can check the existence of a PC FEV $s$-$t$ trail in $G$, and find a minimum weight PC FEV $s$-$t$ trail if it exists, in polynomial time.
\end{proof}

\begin{corollary}\label{corollary: sufficientCondition}
Given a $c$-arc-colored digraph $G$, we can check whether $G$ is PC trail-connected in polynomial time.
\end{corollary}
\begin{proof}
To decide if a given $c$-arc-colored digraph $G$ is PC trail-connected, we just need to check the existence of a PC trail between any pair of arcs $u_1v_1, u_2v_2$ in $G$. Let $H$ be the arc-colored digraph we obtain by deleting from $G$ all outgoing arcs of $u_1$ except $u_1v_1$, and all incoming arcs of $v_2$ except $u_2v_2$. Then a PC trail from $u_1v_1$ to $u_2v_2$ in $G$ corresponds to a PC FEV $u_1$-$v_2$ trail in $H$. As proved in Lemma \ref{lemma:FEVtrail}, checking the existence of a PC FEV $u_1$-$v_2$ trail in $H$ can be done in polynomial time. Note that there are at most $|A(G)|^2$ different arc pairs in $G$, therefore we can decide whether $G$ is PC trail-connected in polynomial time.
\end{proof}

\noindent\textbf{Proof of Theorem \ref{theorem: polynomial}:}
First check whether $G$ is PC trail-connected, which can be done in polynomial time by Corollary \ref{corollary: sufficientCondition}. If $G$ is not PC trail-connected, then there is no solution for CPP-2ACD on $G$, according to Lemma \ref{lemma: solutionCondition2}.

Now assume that $G$ is PC trail-connected. Then by Theorem \ref{theorem:EulerCondition}, we just need to add copies of some arcs in $G$ to make it color-balanced. To decide which set of arcs to be added copies of, we construct, in polynomial time, an undirected weighted complete bipartite graph $H$.

We build the graph $H$ in the following way. For the given graph $G$, define $\theta_{i}^{-}(v)$ and $\theta_{i}^{+}(v)$ for each $v\in V(G)$ and $i\in[2]$ as
$$
 \theta_{i}^{+}(v) = \max\{0, d^{-}_{3-i}(v)-d^{+}_{i}(v)\}, \theta_{i}^{-}(v) = \max\{0, d^{+}_{3-i}(v)-d^{-}_{i}(v)\}.
$$
Note that we need to add at least $\theta_{i}^{+}(v)$ ($\theta_{i}^{-}(v)$) copies of outgoing (incoming) arcs of $v$ colored $i$ in order to make $v$ color-balanced.

Let $X^{+}_{i}(v)$ ($X^{-}_{i}(v)$, respectively) be a vertex set of size $\theta^{+}_{i}(v)$ ($\theta^{-}_{i}(v)$, respectively) for each $v\in V(G)$ and $i\in[2]$. Note that at most one of $\theta^{+}_{i}(v)$ and $\theta^{-}_{3-i}(v)$ is not zero, for $i\in [2]$. Define $$X^{+}_{1} = \bigcup_{v\in V(G)} X^{+}_{1}(v), X^{+}_{2} = \bigcup_{v\in V(G)} X^{+}_{2}(v),$$ $$X^{-}_{1} = \bigcup_{v\in V(G)} X^{-}_{1}(v), X^{-}_{2} = \bigcup_{v\in V(G)} X^{-}_{2}(v).$$ Denote $X^{+} = \bigcup_{v\in V(G), i\in[2]} X^{+}_{i}(v)$, and $X^{-} = \bigcup_{v\in V(G), i\in[2]} X^{-}_{i}(v)$, and let $V(H) = X^{+} \cup X^{-}$. Both $X^{+}$ and $X^{-}$ are independent sets in $H$. 

Add an edge $xy$ between each pair of vertices $x\in X^{+}_{i}(u)$ and $y\in X^{-}_{j}(v)$, for any $i, j\in[2]$ and $u, v\in V(G)$, where we allow $i = j$ and/or $u=v$. Set the weight $\omega(xy)$ to be the minimum weight of a PC FEV $u$-$v$ trail in $G$ with starting color $i$ and ending color $j$. Such a PC $u$-$v$ trail in $G$ can be computed in the following way. Just delete all outgoing arcs of $u$ with color $3-i$ and all incoming arcs of $v$ with color $3-j$ and then compute the minimum weight PC FEV $u$-$v$ trail, which can be done in polynomial time by Lemma \ref{lemma:FEVtrail}. Note that $H$ is a complete bipartite graph as $G$ is PC trail-connected. This completes the description of $H$.

For an arbitrary arc $uv \in V(G)$, without loss of generality, assume it has color $1$, then it is counted once positively in $\Sigma_{u\in V(G)} (d_{1}^{+}(u)-d_{2}^{-}(u))$ and once negatively in $\Sigma_{v\in V(G)} (d_{2}^{+}(v)-d_{1}^{-}(v))$. It follows that
\begin{equation}\label{equation1}
\Sigma_{v\in V(G)}(d_{1}^{+}(v)-d_{2}^{-}(v)) + \Sigma_{v\in V(G)}(d_{2}^{+}(v)-d_{1}^{-}(v))=0.
\end{equation}

Define $$V_1=\{v\in V(G)| d_{1}^{+}(v)-d_{2}^{-}(v)>0\}, V_2=\{v\in V(G)| d_{1}^{+}(v)-d_{2}^{-}(v)<0\},$$  $$V_3=\{v\in V(G)| d_{2}^{+}(v)-d_{1}^{-}(v)>0\}, V_4=\{v\in V(G)| d_{2}^{+}(v)-d_{1}^{-}(v)<0\}.$$ Note that $V_1 \cap V_2 = \emptyset$, and $V_3 \cap V_4 = \emptyset$, therefore (\ref{equation1}) is equivalent to

\begin{IEEEeqnarray*}{rCl}\label{equation2}
\sum_{v\in V_1}(d_{1}^{+}(v)-d_{2}^{-}(v)) + \sum_{v\in V_3}(d_{2}^{+}(v)-d_{1}^{-}(v))&=& \IEEEyesnumber\\
-(\sum_{v\in V_2}(d_{1}^{+}(v)-d_{2}^{-}(v)) + \sum_{v\in V_4}(d_{2}^{+}(v)-d_{1}^{-}(v)))&=&\\
 \sum_{v\in V_2}(d_{2}^{-}(v)-d_{1}^{+}(v)) + \sum_{v\in V_4}(d_{1}^{-}(v)-d_{2}^{+}(v)).
\end{IEEEeqnarray*}

By the definitions of $\theta_{i}^{+}(v), \theta_{i}^{-}(v)$, $X^{+}_{i}(v)$ and $X^{-}_{i}(v)$ for $i\in [2]$, it follows that $$|X^{+}_{1}| = \Sigma_{v\in V_2}\theta_{1}^{+}(v) = \Sigma_{v\in V_2}(d_{2}^{-}(v)-d_{1}^{+}(v)),$$ $$ |X^{-}_{1}| = \Sigma_{v\in V_3}\theta_{1}^{-}(v) = \Sigma_{v\in V_3}(d_{2}^{+}(v)-d_{1}^{-}(v)),$$ $$|X^{+}_{2}| = \Sigma_{v\in V_4}\theta_{2}^{+}(v) = \Sigma_{v\in V_4}(d_{1}^{-}(v)-d_{2}^{+}(v)),$$ $$|X^{-}_{2}|= \Sigma_{v\in V_1}\theta_{2}^{-}(v) = \Sigma_{v\in V_1}(d_{1}^{+}(v)-d_{2}^{-}(v)).$$

Thus (\ref{equation2}) is equivalent to $|X^{+}|=|X^{+}_{1}|+|X^{+}_{2}|=|X^{-}_{1}|+|X^{-}_{2}|=|X^{-}|$. So there must exist a perfect matching in $H$.

Let $e=xy$ be an edge in $H$, with $x\in X^{+}_{i}(u)$ and $y\in X^{-}_{j}(v)$ where we may have $u=v$ and/or $i=j$. An \textit{$e$-trail} $T$ is a minimum weight PC FEV $u$-$v$ trail starting with color $i$ and ending with color $j$. Adding an $e$-trail $T$ to $G$ is to add a copy of each arc which appears in $T$.

%Let $\cal M$ be a perfect matching of $H$ with minimum weight, we show that the weight of $\cal M$ plus the weight of all arcs in $G$ is the weight of an optimal solution for CPP-2ACD on $G$.

We first give a claim about the effects of adding an $e$-trail to $G$.

\medskip
\noindent\textbf{Claim 1:} Let $e=xy$ be an edge in $H$. Adding an $e$-trail $T$ to $G$ has the following effects:
\begin{enumerate}
 \item For any vertex $v\in V(G)$, $d^{+}(v)-d^{-}(v)$ changes if and only if $T$ is open and $v$ is either its first or last vertex.
 \item If $v\in V(G)$ is neither the first nor the last vertex of $T$, then both $d^{+}_{1}(v)-d^{-}_{2}(v)$ and $d^{+}_{2}(v)-d^{-}_{1}(v)$ do not change.
 \item If $T$ is a closed trail starting and ending at $v$, let $i$ ($j$, respectively) be the color of its first (last, respectively) arc. Then $d^{+}_{i}(v)-d^{-}_{3-i}(v)$ increases by 1 if $3-i\neq j$, and it does not change if $3-i= j$; $d^{+}_{3-i}(v)-d^{-}_{i}(v)$ decreases by 1 if $i= j$, and it does not change if $i \neq j$.
 \item If $T$ is open, let $u$ be its first vertex, and $i$ the color of the first arc, then $d^{+}_{i}(u)-d^{-}_{j}(u)$ increases by 1, for any $j \in [2]$.
 Let $v$ be its last vertex, and $j$ the color of the last arc, then $d^{+}_{i}(v)-d^{-}_{j}(v)$ decreases by 1, for any $i \in [2]$.
\end{enumerate}
\noindent\textbf{Proof}. Just observe that each $e$-trail is a PC FEV trail, in which there is no repetition of any arc.\qed

%We now prove the correspondence between any optimal solution for CPP-2ACD on $G$ and a minimum weight perfect matching of $H$.
\medskip
In the following, we prove that the minimum weight of a perfect matching in $H$ plus the weight of $G$ equals to the weight of an optimal solution for CPP-2ACD on $G$. Moreover, the minimum weight perfect matching of $H$ provides us an optimal way to make $G$ color-balanced.

On the one hand, we prove that given any perfect matching $\cal M$ of $H$, if we add an $e$-trail to $G$ for each edge $e\in \cal M$, we get a PC Euler digraph $G'$ (see Claim 2).  On the other hand, we show that for any PC Euler digraph $G''=(V(G), A(G)\cup W)$ which is obtained by adding arcs in $W$, we can decompose arcs in $W$ into a set $\cal F$ of $e$-trails which corresponds to a perfect matching of $H$ (see Claim 3).

\medskip
\noindent\textbf{Claim 2:}
Let $\cal M$ be any perfect matching of $H$, and $G'$ is the digraph we obtain by adding an $e$-trail to $G$ for each edge $e\in \cal M$. Then $G'$ is PC Euler.

\noindent\textbf{Proof}. To prove that $G'$ is PC Euler, by Theorem \ref{theorem:EulerCondition}, we need to show two conditions: $G'$ is PC trail-connected and $G'$ is color-balanced.

First, for any two arcs $f_1, f_2\in A(G)$, there is a PC trail $T$ from $f_1$ to $f_2$ in $G$, as $G$ is PC trail-connected. Observe that $T$ is also a PC trail from $f_1$ to $f_2$ in $G'$, as $G'$ is obtained by adding copies of some arcs to $G$. Consequently $G'$ is PC trail-connected.

Second, to show that $G'$ is color-balanced, we need to prove $d^{-}_{1, G'}(v) = d^{+}_{2,G'}(v)$ and $d^{-}_{2, G'}(v) = d^{+}_{1,G'}(v)$, for any vertex $v\in V(G')$. If $d^{-}_{i, G}(v) > d^{+}_{3-i,G}(v)$, then $|X^{+}_{3-i}(v)| = d^{-}_{i, G}(v) - d^{+}_{3-i,G}(v)$ and $|X^{-}_{i}(v)| =0$, for $i\in [2]$. Therefore $d^{-}_{i, G}(v) - d^{+}_{3-i,G}(v)$  $e$-trails  starting with an outgoing arc of $v$ colored with $3-i$ are added to $G$ and no $e$-trail ending with an incoming arc of $v$ colored with $i$ is added to $G$.
If $d^{-}_{i, G}(v) \leq d^{+}_{3-i,G}(v)$, then $d^{+}_{3-i, G}(v) - d^{-}_{i,G}(v)$  $e$-trails are added to $G$ ending with an incoming arc of $v$ colored with $i$ and no $e$-trail is added to $G$ starting with an outgoing arc of $v$ colored with $3-i$. Thus by Claim 1, $G'$ is color-balanced. It follows that $G'$ is PC Euler. \qed
\medskip

\noindent\textbf{Claim 3:} Let $G^{*}=(V(G), A(G)\cup W)$ be an optimal solution for CPP-2ACD on $G$ with minimum number of arcs, where $W$ is the set of added arcs. Then we can decompose $W$ into a set $\cal F$ of $e$-trails, such that there is an edge set $M = \{e\in E(H): \mbox{there is an $e$-trail in } \cal F\}$  which is a perfect matching in $H$.

\noindent\textbf{Proof}. Note that $G^{*}$ is color-balanced, hence in the graph induced by $W$, we have for any $v\in V(W)$ and $i\in [2]$
$$d^{+}_{3-i,W}(v) - d^{-}_{i,W}(v) = d^{-}_{i,G}(v) - d^{+}_{3-i,G}(v) =\theta^{+}_{3-i}(v) = |X^{+}_{3-i}(v)|,$$ $$d^{-}_{3-i,W}(v) - d^{+}_{i,W}(v) = d^{+}_{i,G}(v) - d^{-}_{3-i,G}(v)= \theta^{-}_{3-i}(v)= |X^{-}_{3-i}(v)|.$$ 

Define $W^{+}_{i} = \{u\in V(W): d^{+}_{i,W}(u)>0, d^{-}_{3-i,W}(u)=0 \}$ and $W^{-}_{j} = \{u\in V(W): d^{-}_{j,W}(v)>0, d^{+}_{3-j,W}(v)=0 \}$  for $i,j\in [2]$.
Observe that for any two vertices $u\in W^{+}_{i}$ and $v\in W^{-}_{j}$, if there is a PC FEV $u$-$v$ trail in $W$ starting with color $i$, and ending with $j$, it must be one with minimum weight in $G$. Otherwise, consider the graph $G^{*'} = G^{*} - T_1 + T_2$, where $T_1\in W, T_2\in G$ both are PC FEV $u$-$v$ trails starting with color $i$, and ending with color $j$, but $T_2$ has smaller weight than $T_1$. Note that $G^{*'}$ is also a PC Euler, as it is PC trail-connected and color-balanced, consequently $G^{*'}$ is a better solution than $G^{*}$, which is a contradiction to the assumption that $G^{*}$ is an optimal solution.

Observe that as $G^{*}$ is an optimal solution with minimum number of arcs for CPP-ACD on $G$, there is no PC cycle in $W$. Otherwise we may delete such PC cycles in $W$ from $G^{*}$, and get a better solution for CPP-ACD on $G$, a contradiction. It follows that we may decompose $W$ into a set $\cal F$ of PC FEV $u$-$v$ trails such that $u\in W^{+}_{i}$ and $v\in W^{-}_{j}$ for $i,j\in [2]$ where we allow $u=v$ and/or $i=j$. By the above arguments, we know each PC FEV $u$-$v$ trail has minimum possible weight, thus corresponds to an edge in $H$ between $X^{+}(u)$ and $X^{-}(v)$.
For any edge set $M = \{e'\in E(X^{+}_{i}(u),X^{-}_{j}(v)):$ there is an $e$-trail in $\cal F$ which starts at $u$ with color $i$ and ends at $v$ with color $j$, where we allow $u=v$ and/or $i=j \}$, there are $|X^{+}_{i}(v)|$ edges in $M$ that are incident with vertices in $X^{+}_{i}(v)$ and
$|X^{-}_{i}(v)|$ edges in $M$ that are incident with vertices in $X^{-}_{i}(v)$ for each $i\in [2]$. Moreover, $H$ is a complete bipartite graph. Consequently, we may choose $M$ such that any two edges in $M$ are not incident with each other.
And thus $M$ is a perfect matching in $H$.\qed

\medskip
 By Claim 2, for any perfect matching $\cal M$ in $H$, adding an $e$-trail to $G$ for each edge $e\in \cal M$ results in a PC Euler graph. Let $G''=(V(G), A(G)\cup W)$ be a PC Euler digraph with minimum weight and minimum number of arcs.
  By Claim 3,  $W$ corresponds to a perfect matching in $H$. Thus for a minimum weight perfect matching $\cal M^{*}$ in $H$, adding an $e$-trail to $G$ for each edge $e\in \cal M^{*}$ results in a minimum weight PC Euler graph $G^{*}$. So to solve CPP-2ACD on $G$, it suffices to compute a minimum weight perfect matching in $H$, which can be done in $O(|V(H)|^3)$ by the Hungarian method, thus we can solve CPP-2ACD on $G$ in polynomial time.\qed

\section{Discussions}
In this paper, we generalize the polynomial time solvability of Chinese Postman Problem on 2-edge-colored graphs to 2-arc-colored weighted digraphs. First, we give a necessary and sufficient condition for a $c$-arc-colored digraph to have a solution for CPP-ACD. We then provide a characterization of PC Euler 2-arc-colored digraphs and show how to find a PC Euler trail of such a graph in polynomial time. Using the characterization, we further show how to solve CPP-2ACD on a given 2-arc-colored digraph in polynomial time.

It would be interesting to see whether our result can be further extended to $c$-arc-colored digraphs for $c>2$. We leave this as an open problem. Note that being PC trail-connected and color-balanced is not sufficient for a $c$-arc-colored digraph to be PC Euler. Figure \ref{counterexample} gives a PC trail-connected color-balanced graph $H$ which is not PC Euler. It is easy to see that $H$ is color-balanced. It is also not hard to see that every arc $f_i$ of the cycle $vu_iu_{i+1}v$ can reach every arc $f_j$ of the cycle $vu_ju_{j+1}v$ using a trail starting at $f_i$ and ending at $f_j$ (in the case $\{i,j\} =\{1,3\}$, we will have to use arcs of the cycle $vu_5u_6v$). To see that $H$ has no PC Euler trail, observe that we cannot switch between cycles $vu_1u_{2}v$ and $vu_3u_{4}v$ twice as each switch will require arcs of the cycle $vu_5u_6v$.

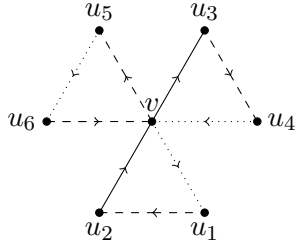
\begin{figure}\label{counterexample}
\centering
\begin{tikzpicture}[scale = 0.7]

\draw [](-1,1.732)node[above]{$u_5$};
\draw [](-2,0)node[left]{$u_6$};
\draw [](2,0)node[right]{$u_4$};
\draw [](1,1.732)node[above]{$u_3$};
\draw [](0,0)node[above]{$v$};
\draw [](-1,-1.732)node[below]{$u_2$};
\draw [](1,-1.732)node[below]{$u_1$};

\draw [,dashed](-1,-1.732)--(0,-1.732);
\draw [<-,dashed](0,-1.732)--(1,-1.732);

\draw [<-,dashed](1.5,0.866)--(1,1.732);
\draw [dashed](2,0)--(1.5,0.866);

\draw [dashed](-1/2,0.866)--(-1,1.732);
\draw [->,dashed](0,0)--(-1/2,0.866);

\draw [dashed](-2,0)--(0,0);
\draw [->,dashed](-2,0)--(-1,0);

\draw [dotted](-1.5,0.866)--(-2,0);
\draw [->,dotted](-1,1.732)--(-1.5,0.866);

\draw [dotted](1,0)--(0,0);
\draw [->,dotted](2,0)--(1,0);

\draw [dotted](1/2,-0.866)--(1,-1.732);
\draw [->,dotted](0,0)--(1/2,-0.866);

\draw [->,](-1,-1.732)--(-1/2,-0.866);
\draw [](-1/2,-0.866)--(0,0);

\draw [](1/2,0.866)--(1,1.732);
\draw [->,](0,0)--(1/2,0.866);

\draw [](-2,0)[fill]circle[radius=0.07];
\draw [](2,0)[fill]circle[radius=0.07];
\draw [](1,-1.732)[fill]circle[radius=0.07];
\draw [](1,1.732)[fill]circle[radius=0.07];
\draw [](-1,-1.732)[fill]circle[radius=0.07];
\draw [](-1,1.732)[fill]circle[radius=0.07];
\draw [](0,0)[fill]circle[radius=0.07];

\end{tikzpicture}
\caption{The graph $H$.}
 \label{fig:lemmaProof1}
\end{figure}


\begin{thebibliography}{10}

\bibitem{abouelaoualim2010cycles}
A~Abouelaoualim, K~Ch Das, W~Fernandez de~la Vega, Marek Karpinski, Yannis
  Manoussakis, CA~Martinhon, and Rachid Saad.
\newblock Cycles and paths in edge-colored graphs with given degrees.
\newblock {\em Journal of Graph Theory}, 64(1):63--86, 2010.

\bibitem{ahuja2010algorithms}
Sandeep~Kour Ahuja.
\newblock {\em Algorithms for routing and channel assignment in wireless
  infrastructure networks}.
\newblock PhD thesis, The University of Arizona, 2010.

\bibitem{amaldi2011minimum}
Edoardo Amaldi, Giulia Galbiati, and Francesco Maffioli.
\newblock On minimum reload cost paths, tours, and flows.
\newblock {\em Networks}, 57(3):254--260, 2011.

\bibitem{bang2009digraphs}
J{\o}rgen Bang-Jensen and Gregory~Z Gutin.
\newblock {\em Digraphs: theory, algorithms and applications}.
\newblock Springer, 2 edition, 2009.

\bibitem{bernhard2008combinatorial}
Korte Bernhard and J~Vygen.
\newblock Combinatorial optimization: Theory and algorithms.
\newblock {\em Springer, Third Edition, 2005.}, 2008.

\bibitem{van2014complexity}
Ren{\'e}~van Bevern, Rolf Niedermeier, Manuel Sorge, and Mathias Weller.
\newblock The complexity of arc routing problems.
\newblock In {\em Arc Routing: Problems, Methods, and Applications}. SIAM,
  2014.

\bibitem{bollobas2013modern}
B{\'e}la Bollob{\'a}s.
\newblock {\em Modern graph theory}, volume 184.
\newblock Springer, 2013.

\bibitem{edmonds1973matching}
Jack Edmonds and Ellis~L Johnson.
\newblock Matching, euler tours and the chinese postman.
\newblock {\em Mathematical programming}, 5(1):88--124, 1973.

\bibitem{fernandes2009minimum}
Cristina~G Fernandes, Orlando Lee, and Yoshiko Wakabayashi.
\newblock Minimum cycle cover and chinese postman problems on mixed graphs with
  bounded tree-width.
\newblock {\em Discrete Applied Mathematics}, 157(2):272--279, 2009.

\bibitem{fujita2011properly}
Shinya Fujita and Colton Magnant.
\newblock Properly colored paths and cycles.
\newblock {\em Discrete Applied Mathematics}, 159(14):1391--1397, 2011.

\bibitem{galbiati2014minimum}
Giulia Galbiati, Stefano Gualandi, and Francesco Maffioli.
\newblock On minimum reload cost cycle cover.
\newblock {\em Discrete Applied Mathematics}, 164:112--120, 2014.

\bibitem{gamvros2006satellite}
Ioannis Gamvros.
\newblock {\em Satellite network, design, optimization, and management}.
\newblock PhD thesis, University of Maryland, 2006.

\bibitem{gourves2013complexity}
Laurent Gourv{\`e}s, Adria Lyra, Carlos~A Martinhon, and J{\'e}r{\^o}me Monnot.
\newblock Complexity of trails, paths and circuits in arc-colored digraphs.
\newblock {\em Discrete Applied Mathematics}, 161(6):819--828, 2013.

\bibitem{gutin2015parameterized:2}
Gregory Gutin, Mark Jones, and Bin Sheng.
\newblock Parameterized complexity of the k-arc chinese postman problem.
\newblock In {\em European Symposium on Algorithms}, pages 530--541. Springer,
  2014.

\bibitem{gutin2014parameterized:3}
Gregory Gutin, Mark Jones, Bin Sheng, and Magnus Wahlstr{\"o}m.
\newblock Parameterized directed k-chinese postman problem and k arc-disjoint
  cycles problem on euler digraphs.
\newblock In {\em International Workshop on Graph-Theoretic Concepts in
  Computer Science}, pages 250--262. Springer, 2014.

\bibitem{gutin2017chinese}
Gregory Gutin, Mark Jones, Bin Sheng, Magnus Wahlstr{\"o}m, and Anders Yeo.
\newblock Chinese postman problem on edge-colored multigraphs.
\newblock {\em Discrete Applied Mathematics}, 217:196--202, 2017.

\bibitem{gutin2015structural}
Gregory Gutin, Mark Jones, and Magnus Wahlstr{\"o}m.
\newblock Structural parameterizations of the mixed chinese postman problem.
\newblock In {\em Algorithms-ESA 2015}, pages 668--679. Springer, 2015.

\bibitem{gutin2013parameterized}
Gregory Gutin, Gabriele Muciaccia, and Anders Yeo.
\newblock Parameterized complexity of k-chinese postman problem.
\newblock {\em Theoretical Computer Science}, 513:124--128, 2013.

\bibitem{gutin2017odd}
Gregory Gutin, Bin Sheng, and Magnus Wahlstr{\"{o}}m.
\newblock Odd properly colored cycles in edge-colored graphs.
\newblock {\em Discrete Mathematics}, 340(4):817--821, 2017.

\bibitem{gutin1998note}
Gregory Gutin, Benjamin Sudakov, and Anders Yeo.
\newblock Note on alternating directed cycles.
\newblock {\em Discrete Mathematics}, 191(1-3):101--107, 1998.

\bibitem{kotzig1968moves}
Anton Kotzig.
\newblock Moves without forbidden transitions in a graph.
\newblock {\em Matematick{\`y} {\v{c}}asopis}, 18(1):76--80, 1968.

\bibitem{lo2014dirac}
Allan Lo.
\newblock A dirac type condition for properly coloured paths and cycles.
\newblock {\em Journal of Graph Theory}, 76(1):60--87, 2014.

\bibitem{lo2014edge}
Allan Lo.
\newblock An edge-colored version of dirac's theorem.
\newblock {\em SIAM Journal on Discrete Mathematics}, 28(1):18--36, 2014.

\bibitem{papadimitriou1976complexity}
Christos~H Papadimitriou.
\newblock On the complexity of edge traversing.
\newblock {\em Journal of the ACM (JACM)}, 23(3):544--554, 1976.

\bibitem{pevzner18computational}
PA~Pevzner.
\newblock Computational molecular biology: an algorithmic approach.
\newblock {\em Cambridge, Mass.: MIT Press}, 18(3):1, 2000.

\bibitem{sankararaman2014channel}
Swaminathan Sankararaman, Alon Efrat, Srinivasan Ramasubramanian, and Pankaj~K
  Agarwal.
\newblock On channel-discontinuity-constraint routing in wireless networks.
\newblock {\em Ad hoc networks}, 13:153--169, 2014.

\bibitem{szachniuk2014orderly}
Marta Szachniuk, Maria~Cristina De~Cola, Giovanni Felici, and Jacek Blazewicz.
\newblock The orderly colored longest path problem-a survey of applications and
  new algorithms.
\newblock {\em RAIRO-Operations Research}, 48(1):25--51, 2014.

\bibitem{szachniuk2009assignment}
Marta Szachniuk, Mariusz Popenda, Ryszard~W Adamiak, and Jacek Blazewicz.
\newblock An assignment walk through {3D NMR} spectrum.
\newblock In {\em Computational Intelligence in Bioinformatics and
  Computational Biology, 2009. CIBCB'09. IEEE Symposium on}, pages 215--219.
  IEEE, 2009.

\bibitem{thomassen1997complexity}
Carsten Thomassen.
\newblock On the complexity of finding a minimum cycle cover of a graph.
\newblock {\em SIAM Journal on Computing}, 26(3):675--677, 1997.

\bibitem{tseng2010obstacle}
I-Lun Tseng, Huan-Wen Chen, and Che-I Lee.
\newblock Obstacle-aware longest-path routing with parallel milp solvers.
\newblock In {\em World Congress on Engineering and Computer Science (WCECS)},
  volume~2. Citeseer, 2010.

\bibitem{xu2010using}
Haiyan Xu, D~Marc Kilgour, Keith~W Hipel, and Graeme Kemkes.
\newblock Using matrices to link conflict evolution and resolution in a graph
  model.
\newblock {\em European Journal of Operational Research}, 207(1):318--329,
  2010.

\bibitem{xu2009matrix}
Haiyan Xu, Kevin~W Li, D~Marc Kilgour, and Keith~W Hipel.
\newblock A matrix-based approach to searching colored paths in a weighted
  colored multidigraph.
\newblock {\em Applied Mathematics and Computation}, 215(1):353--366, 2009.

\bibitem{yeo1997note}
Anders Yeo.
\newblock A note on alternating cycles in edge-coloured graphs.
\newblock {\em Journal of Combinatorial Theory, Series B}, 69(2):222--225,
  1997.

\end{thebibliography}
\end{document}